\documentclass[journal,draftclsnofoot,onecolumn,12pt,twoside]{IEEEtran}
\usepackage{amsmath,amsfonts}
\usepackage{algorithmic}
\usepackage{algorithm}
\usepackage{amssymb}

\usepackage{color}

\usepackage{float} 

\usepackage{array}
\usepackage[caption=false,font=normalsize,labelfont=sf,textfont=sf]{subfig}
\usepackage{textcomp}
\usepackage{stfloats}
\usepackage{url}
\usepackage{bm}
\usepackage{setspace}
\usepackage{verbatim}
\usepackage{multirow} 
\usepackage{ntheorem}
\usepackage{graphicx}
\usepackage{cite}
\usepackage{amsmath}

\allowdisplaybreaks[4]

\theorembodyfont{\upshape}
\theoremheaderfont{\it}
\qedsymbol{\square}
\theoremseparator{:}
\newtheorem{theorem}{\indent Theorem}
\newtheorem{lemma}{\indent Lemma}
\newtheorem*{proof}{\indent Proof}
\newtheorem{remark}{\indent Remark}

\makeatletter

\newcommand{\Rmnum}[1]{\expandafter\@slowromancap\romannumeral #1@}
\makeatother
\begin{document}

\title{\huge GoMORE: Global Model Reuse for Resource-Constrained Wireless Federated Learning}

\author{Jiacheng~Yao,
		Zhaohui~Yang,
		Wei~Xu,~\IEEEmembership{Senior~Member,~IEEE,}
		Mingzhe~Chen,
		and Dusit~Niyato,~\IEEEmembership{Fellow,~IEEE}

\thanks{Jiacheng Yao and Wei Xu are with the National Mobile Communications Research Laboratory (NCRL), Southeast University, Nanjing 210096, China (\{jcyao, wxu\}@seu.edu.cn).}
\thanks{Zhaohui Yang is with the Zhejiang Lab, Hangzhou 311121, China, and also with the College of Information Science and Electronic Engineering, Zhejiang University, Hangzhou, Zhejiang 310027, China (yang\_zhaohui@zju.edu.cn).}
\thanks{Mingzhe Chen is with the Department of Electrical and Computer Engineering and Institute for Data Science and Computing, University of Miami, Coral Gables, FL 33146 USA (mingzhe.chen@miami.edu).}
\thanks{Dusit Niyato is with the School of Computer Science and Engineering, Nanyang Technological University, Singapore 308232 (dniyato@ntu.edu.sg).}
}

%

\maketitle

\begin{abstract}
Due to the dynamics of wireless environment and limited bandwidth, wireless federated learning (FL) is challenged by frequent transmission errors and incomplete aggregation from devices. In order to overcome these challenges, we propose a \underline{G}l\underline{o}bal \underline{MO}del \underline{RE}use strategy (GoMORE) that reuses the outdated global model to replace the local model parameters once a transmission error occurs. We analytically prove that the proposed GoMORE is strictly superior over the existing strategy, especially at low signal-to-noise ratios (SNRs). In addition, based on the derived expression of weight divergence, we further optimize the number of participating devices in the model aggregation to maximize the FL performance with limited communication resources. Numerical results verify that the proposed GoMORE successfully approaches the performance upper bound by an ideal transmission. It also mitigates the negative impact of non-independent and non-identically distributed (non-IID) data while achieving over 5 dB reduction in energy consumption.
\end{abstract}

\begin{IEEEkeywords}
Federated learning (FL), unreliable communication, limited bandwidth, partial aggregation
\end{IEEEkeywords}

\section{Introduction}
\IEEEPARstart{D}{ue} to the rapid growth of data traffic in beyond fifth-generation (B5G) networks, federated learning (FL) is regarded as an important enabling key technique for wireless networks to support various data-driven intelligent applications \cite{xu,push}. Specifically, the distributed devices, coordinated by a parameter server (PS), collaboratively train a shared learning model by exploiting local private data with marginal privacy leaks~\cite{fl}. Recently, state-of-art works have considered the specific deployments of FL algorithm over wireless networks \cite{energy,ajoint}, where the devices and the PS communicate through wireless links. However, the communication resource is always physically limited and the dynamics of wireless environment presents, both of which have become the most significant bottlenecks limiting the learning performance of wireless FL. 

Dynamic wireless channel conditions usually lead to unignorable packet loss due to transmission errors, which hinders the convergence of FL algorithms and deteriorates the learning performance. In \cite{ajoint}, the impact of unreliable communication on the convergence performance of an FL algorithm was analyzed and communication resource allocation was optimized based on the derived convergence bounds.
In \cite{wirelessfl}, the convergence performance in resource-constrained cellular wireless networks was analyzed.
In addition, the effects of model pruning and packet error on the convergence were quantified in \cite{ni}, which was further used to seek a balance between the communication and learning.
To combat with unreliable communication, the user datagram protocol (UDP) was adopted in \cite{ye} with a robust training algorithm, which retains the same asymptotic convergence rate as that with error-free communications. Moreover, in \cite{update}, outdated local updates of each device was additionally exploited to replace local updates with transmission errors. 

On the other hand, limited spectrum resource places substantial restrictions on the number of devices that can be simultaneously involved in FL model training of each communication round, which also harms the learning process. Existing works tried to devise device scheduling strategies at the PS to purse better learning performance, such as the importance-based sampling \cite{importance} and channel conditions-based scheduling \cite{channel}. However, it is worth pointing out that there is a trade-off between the two challenges mentioned above. For a given bandwidth, more participating devices results in less bandwidth occupied by each device, which in turn brings more transmission errors, and vice versa. To the best of our knowledge, few efforts have been endeavored to achieve a balance between higher probability  of successful transmission and more participating devices under the practical constraints of limited communication resource.

In this paper, we characterize the coupled impacts of unreliable communication and limited bandwidth for wireless FL by deriving theoretical bounds of the weight divergence. A \underline{G}l\underline{o}bal \underline{MO}del \underline{RE}use strategy, namely GoMORE, is proposed to combat the transmission errors. We prove in theory that the GoMORE always outperforms the existing direct discarding strategy. Moreover, based on the derived analytical results, we further optimize the number of participating devices, which balances between transmission errors and the number of participating devices. Numerical experiments are conducted to verify the superiority of GoMORE especially in the case with low signal-to-noise ratios (SNRs) and non-independent and non-identically distributed (non-IID) local data sets. 

The rest of this paper is organized as follows. Section II formulates the system model and elaborates on the proposed strategy. In Section III, we derive the analytical performance of the proposed strategy and optimize the number of participating devices. Simulation results and concluding remarks are given in Sections IV and V, respectively.
\vspace{-0.1cm}
\section{System Model}
We consider a typical wireless FL system, where $K$ distributed devices communicate with a PS via wireless channels. By exploiting the non-IID local data sets $\left\{\mathcal{D}_k\right\}_{k=1}^K$ owned by the devices and through wireless communications between the devices and PS, we aim to cooperatively train the global model parameters via FL. The FL is modelled as optimizing the global parameters, $\bm{w}$, to minimize the loss function
\begin{align}\label{e1}
F(\bm{w})=\sum_{k=1}^K \frac{D_k}{\sum_{i=1}^K D_i} F_k(\bm{w},\mathcal{D}_k),
\end{align}
\noindent where $D_k$ is the size of the $k$-th local data set $\mathcal{D}_k$ and $F_k(\cdot)$ is the local loss function at device $k$. Without loss of generality, we assume that the size of all local data sets are the same, i.e., $D_1=\cdots=D_K=D$. A typical local loss function at device $k$ is defined as
\begin{align}
 F_k(\bm{w},\mathcal{D}_k)=\frac{1}{D}\sum_{\bm{x}\in \mathcal{D}_k} \mathcal{L}(\bm{w},\bm{x}),
\end{align}
where $\bm{x}$ is data selected from $\mathcal{D}_k$ and $\mathcal{L}(\bm{w},\bm{x})$ is the loss function for data sample $\bm{x}$. Note that, the data of different devices are usually non-IID distributed in practice, which is different from the centralized learning paradigms and brings additional challenges.
\vspace{-0.15cm}
\subsection{Federated Learning Model}
Next, we introduce the typical FL algorithm under ideal transmission. Specifically, the $m$-th round of the FL algorithm is composed of the following three steps.
\subsubsection{Broadcasting}  The PS broadcasts the latest global parameters $\bm{w}_{m}$ to all devices.
\subsubsection{Local computation} Each device receives the global parameter $\bm{w}_{m}$ and initializes its local model as $\bm{w}_{m}$. Then, each device $k$ runs $T$ local training epochs before obtaining an updated model. At the $t$-th epoch for $1\leq t\leq T$, it follows:
\begin{align} \label{e3}
\bm{w}_{mT+t}^k = \bm{w}_{mT+t-1}^k -\eta_m \nabla F_k\left(\bm{w}_{mT+t-1}^k,\bm{\xi}_{mT+t-1}^k \right),
\end{align}
where the subscript $mT+t$ represents the cumulative number of local epochs, the superscript $k$ corresponds to the $k$-th device, $\eta_m$ is the learning rate chosen at the $m$-th round, and $\bm{\xi}_{mT+t-1}^k$ is the batch of samples with size $b$. According to the definition, we have $\bm{w}_{mT}^k=\bm{w}_m$.
\subsubsection{Uplink transmission and aggregation} 
All the devices transmit their updated local models, $\bm{w}_{(m+1)T}^k$, to the PS.
Upon receiving all the local models, PS updates the global model as
\begin{align}\label{e5}
\bm{w}_{m+1}= \frac{1}{K} \sum_{k=1}^K  \bm{w}_{(m+1)T}^k.
\end{align}
Note that the above steps iterate until the FL algorithm converges to a common global model \cite{ajoint}.

\vspace{-0.1cm}
\subsection{Unreliable Transmission Model}
Given that the bandwidth available for FL is always limited, only a subset of $N$ devices can participate in model training per round. We assume that the PS uniformly selects $N$ devices without replacement \cite{importance} and the set of the selected devices at the $m$-th round is denoted by $\mathcal{S}_m$. Then, in the above steps, only the selected devices participate in model aggregation. 

Moreover, over wireless channels, transmission errors occurs during the uplink local FL parameter report. For the downlink transmission, we assume that the broadcast of global model parameters is error-free due to much higher transmit power at the PS \cite{ajoint}.
Specifically, we denote the channel gain between the PS and device $k$ by $h_0 h_k(m)d_k^{-\frac{\alpha}{2}}$, where $h_0$ is the channel gain at the reference distance, $h_k(m)\sim \mathcal{CN}(0,1)$ represents the small-scale fading in the $m$-th round, $d_k$ is the distance between PS and device $k$, and $\alpha$ is the path-loss exponent. Then, the SNR in the $m$-th round follows
\begin{align}
\gamma_k = \frac{P \vert h_0 h_k(m) \vert^2 d_k^{-\alpha}}{B_k N_0},\enspace \forall k,
\end{align}
where $P$ is the transmit power, $B_k$ is the bandwidth allocated to device $k$, and $N_0$ is the noise power density. Without loss of generality, we assume that all the selected devices share the bandwidth equally, i.e., $B_k=\frac{B}{N}$. Then, the local updates are transmitted at the fixed rate $\frac{B}{N}\log(1+\theta)$, where $\theta$ is a chosen constant\footnote{To avoid the impacts of stragglers and accelerate model training, the local updates are transmitted at a fixed rate, like in \cite{wirelessfl}, rather than a dynamic rate based on the SNR levels.}. A reception is assumed error-free if the SNR exceeds $\theta$, and otherwise error occurs \cite{wirelessfl}. Explicitly, the probability of error-free transmission can be evaluated as \cite{zhu}
\begin{align}\label{e7}
p_k= \exp\left(-\frac{B N_0 \theta}{2 N P \vert h_0 \vert^2 d_k^{-\alpha}} \right),\enspace \forall k.
\end{align}
Typically in most existing FL systems \cite{ajoint}, the PS relies on a cyclic redundancy check (CRC) mechanism to check the detected data and directly discards the local updates with errors while not asking for retransmission, referred to as the direct discarding strategy (DDS). Specifically, the aggregation step at the $m$-th round in (\ref{e5}) is rewritten for the wireless FL as
\begin{align}
\tilde{\bm{w}}_{m+1}= \sum_{k\in \mathcal{S}_m} \frac{1}{Np_k} \tilde{\bm{w}}_{(m+1)T}^k,
\end{align}
where $\tilde{\bm{w}}_{m+1}^k$ denotes the received local model and it is a discrete random variable taking the value of $\bm{w}_{m+1}^k$ with probability $p_k$ and the value of $0$ with probability $1-p_k$. In addition, $\tilde{\bm{w}}_{m+1}$ denotes the updated global model based on DDS. It is obvious that this averaging step is unbiased, i.e., $\mathbb{E} \left [ \tilde{\bm{w}}_{m+1}\right ]=\bm{w}_{m+1}$. However, the learning performance of the FL algorithm based on DDS is not guaranteed especially for the non-IID case. This is because the absence of updates from some devices may bias the global model towards the remaining devices, leading to excessively large gap with respect to the optimal direction. As shown in Fig.~1, due to absence of $\bm{w}_{(m+1)T}^1$, the global model gets far from the global minimum.

\begin{figure}[!t]
\centering
\includegraphics[width=3in]{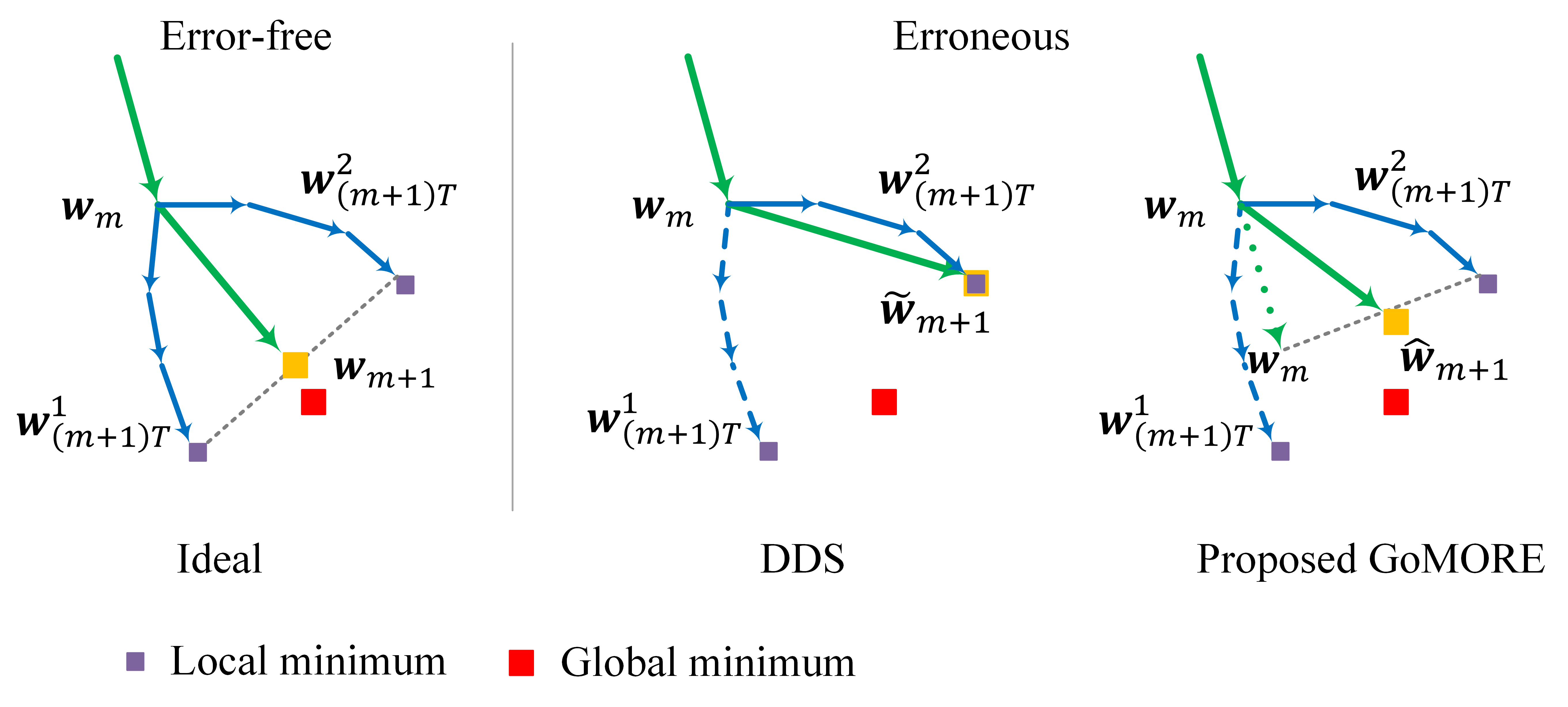}
\caption{ An example reflecting the impact of unreliable communications.}
\vspace{-0.4cm}
\end{figure}

\vspace{-0.1cm}
\subsection{Global Model Reuse (GoMORE) Strategy}

To address the challenge brought by DDS, we propose a novel global model reuse strategy, where \emph{the global model updated at the previous communication round is reused as an alternative of the erroneously received local models}. Concretely, in the proposed GoMORE, if the transmission is error-free, the model parameters are used in the aggregation as usual, and conversely, if the transmitted model parameters have transmission errors, the outdated global model is used to replace the erroneous one. Hence, the aggregation step in $m$-th round becomes
\begin{align}\label{agg}
\hat{\bm{w}}_{m+1}= \frac{1}{N}\sum_{k\in \mathcal{S}_m} \hat{\bm{w}}_{(m+1)T}^k,
\end{align}
where $\hat{\bm{w}}_{(m+1)T}^k$ is defined by
\begin{align}\label{def1}
\hat{\bm{w}}_{(m+1)T}^k\triangleq \left \{  \begin{array}{cc}
\bm{w}_{(m+1)T}^k, & \text{error-free}\enspace \text{w.p.}\enspace p_k,\\ 
\bm{w}_{m},& \text{erroneous}\enspace \text{w.p.}\enspace 1-p_k.
\end{array} \right.
\end{align}
In GoMORE, even though there is packet loss, the alternative global model is able to help maintain the original direction as much as possible without excessive bias, as depicted in Fig.~1.

\section{Performance Analysis and Optimization}
In this section, we analyze the performance of the proposed GoMORE and compare it with DDS. Based on the theoretical results, we further optimize the number of participating devices under limited resource constraints and a stringent delay requirement to improve the performance of GoMORE.
\vspace{-0.2cm}
\subsection{Performance Analysis and Comparison}
To capture the impacts of both unreliable transmission and limited bandwidth, we utilize the expected weight divergence with respect to $\hat{\bm{w}}_{m+1}$ and $\bm{w}_{m+1}$, i.e., $\zeta_1=\mathbb{E}\left[\left \Vert \hat{\bm{w}}_{m+1}- \bm{w}_{m+1}\right \Vert^2 \right]$, as the performance metric, which effectively evaluates the accuracy loss caused by non-IID data \cite{mnist}. More intuitively, weight divergence is understood as the second-order moments of the estimation error for the ideal model parameters. For performance comparison, we analyze the metric of  $\zeta_2=\mathbb{E}\left[\left \Vert \tilde{\bm{w}}_{m+1}- \bm{w}_{m+1}\right \Vert^2\right] $ for DDS. To facilitate the performance analysis, we need the following assumptions, which has been widely used in literature, e.g.,~\cite{zhu2}.

\emph{Assumption 1}: The stochastic gradients on random data samples are uniformly bounded by a finite constant $\gamma^2$, i.e., $\mathbb{E}_{\bm{\xi}}\left[\left \Vert \nabla F_k(\bm{w},\bm{\xi}) \right \Vert^2\right]\leq \gamma^2$.

\emph{Assumption 2}: The expected squared norm of the model parameters are uniformly bounded by a finite constant $G^2$, i.e., $\mathbb{E}\left[\left \Vert  \bm{w} \right \Vert^2\right]\leq G^2$.

Under these general assumptions, we characterize the weight divergence, $\zeta_1$ and $\zeta_2$ by the following bounds in \emph{Lemma 1} and \emph{Lemma 2}, respectively.
\begin{lemma}
An upper bound of $\zeta_1$ is expressed as
\begin{align} \label{l1}
\zeta_1 \leq \sum_{k=1}^K \frac{\eta_m^2T^2 \gamma^2}{K}\left( \frac{K-N}{N(K-1)} p_k^2-p_k+1\right)\triangleq \bar{\zeta}_1.
\end{align}
\end{lemma}
\begin{proof}
Please refer to Appendix A. $\hfill \square$
\end{proof}

\begin{lemma}
An upper bound of $\zeta_2$ is expressed as
\begin{align} \label{l2}
\zeta_2 \leq \sum_{k=1}^K \left(\frac{\eta_m^2T^2\gamma^2(K-N)}{KN(K-1)}+ \frac{1-p_k}{Kp_k}G^2\right)\triangleq \bar{\zeta}_2.
\end{align}
\end{lemma}
\begin{proof}
Please refer to Appendix B. $\hfill \square$
\end{proof}

Based on \emph{Lemma 1} and \emph{Lemma 2}, we have the following theorem to compare the DDS and the proposed GoMORE.

\begin{theorem}
With sufficiently small learning rate, i.e., $\eta_m\leq \frac{G}{T\gamma}$, the global model obtained via GoMORE is strictly superior to that obtained via DDS in terms of weight divergence.
\end{theorem}

\begin{proof}
According to (\ref{l1}), (\ref{l2}) and  $\eta_m\leq \frac{G}{T\gamma}$, we have
\begin{align} \label{e13}
\bar{\zeta}_2-\bar{\zeta}_1\geq \sum_{k=1}^K\frac{\eta_m^2 T^2 \gamma^2}{K}\left( \frac{(K-N)(1-p_k^2)}{N(K-1)} +\frac{(1-p_k)^2}{p_k}\right)>0.
\end{align}
It is obvious that $\bar{\zeta}_1<\bar{\zeta}_2$ and we complete the proof. $\hfill \square$
\end{proof}

\begin{remark}
It is worth noting that (\ref{e13}) decreases monotonically as $p_k$ increases and eventually converges to 0. This indicates that the proposed GoMORE has more significant performance advantages over the DDS at low SNR regions and enjoys the same asymptotical upper bound at high SNRs.
\end{remark}

\subsection{Optimization of Device Activation}
Next, we focus on the optimization of the number of participating devices. Let $d$ denote the data size of the model parameters $\bm{w}$ (in bit). To meet the upload time delay requirement, $\tau$, the communication rate should be set as $d/\tau$. According to (\ref{e7}), $p_k$ is evaluated as
$p_k=\exp\left (-\lambda_k \frac{2^{\rho N}-1}{N}\right)$,
where $\lambda_k\triangleq \frac{BN_0}{2P\vert h_0 \vert^2d_k^{-\alpha}}$ and $\rho\triangleq \frac{d}{B\tau}$. Then, it is obvious that the upper bound $\bar{\zeta}_1$ is not a monotonic function with respect to $N$. Therefore, with the constraint of limited bandwidth and a delay requirement, there should exist an optimal number of participating devices. By removing constant terms in (\ref{l1}), we formulate the optimization problem as
\begin{align}\label{e16}
\mathop{\mathrm{minimize}}_{N} &\quad \sum_{k=1}^K\left( \frac{K-N}{(K-1)N} e^{-2\lambda_k \frac{2^{\rho N}-1}{N}}-e^{-\lambda_k \frac{2^{\rho N}-1}{N}}\right)\nonumber \\
\mathrm{subject}\enspace \mathrm{to} &\quad 1\leq N\leq K.
\end{align}
Given that $N$ is a discrete variable, an exhaustive search is useful to find the optimal $N$ by minimizing the upper bound $\bar{\zeta}_1$ with computational complexity of $\mathcal{O}(N)$. In practice, the number of the devices, $N$, is not large and hence an exhaustive search works. It is worth noting that the effects of the learning parameters and the number of participating devices are decoupled and hence the optimized $N$ applies for various learning parameter settings.
We also note that the optimization of $N$ is based on the statistics of the wireless channels and hence is applicable for long-term settings.

To summarize, we conclude the proposed GoMORE strategy in Algorithm \ref{alg1}.

\begin{algorithm}[!t]\small
\caption{The GoMORE strategy for wireless FL}
\begin{algorithmic}[1]  \label{alg1}
\STATE \textbf{Optimize} $N$ according to (\ref{e16}).
\STATE \textbf{Repeat}
\STATE $\quad$PS broadcasts $\bm{w}_m$ to the randomly selected $N$ devices.
\STATE $\quad$Each selected device conducts local computing and transmits $\bm{w}_{(m+1)T}^k$ to PS.
\STATE $\quad$The PS performs model aggregation according to (\ref{agg}).
\STATE \textbf{End until} convergence.
\end{algorithmic} 

\vspace{-0.1cm}
\end{algorithm}
\section{Simulation Results}
In this section, simulation results are provided to verify the effectiveness of the proposed GoMORE. The popular MNIST data set is exploited to train a multi-layer perceptron (MLP) and the label distribution varies over devices to capture the non-IID characteristic. Unless otherwise specified, the parameters are set as: the number of all devices, $K=20$, the number of local epochs, $T=10$, the batch size $b=50$, the learning rate $\eta_m=0.001$, the total bandwidth, $B=1$ MHz, the noise power spectral density, $N_0=-174$ dBm/Hz, the channel gain at reference distance, $h_0=-30$ dBm, and the path loss exponent, $\alpha=2.2$.

In Fig. 2 and Fig. 3, we compare the proposed GoMORE with DDS under different parameter settings. Consistent with our theoretical analysis, GoMORE outperforms DDS for all the tested setups and approaches the performance of ideal transmission with a more than 5 dB reduction in energy consumption. Moreover, compared with the IID data sets, the learning performance is more sensitive to the transmission errors under the non-IID data. Hence, the proposed GoMORE plays a more critical role for the non-IID cases, which avoids excessive bias brought by unexpected errors.

\begin{figure*}[!t]
	\centering
	\begin{minipage}[t]{0.25\linewidth}
		\centering
		\includegraphics[width=1\linewidth]{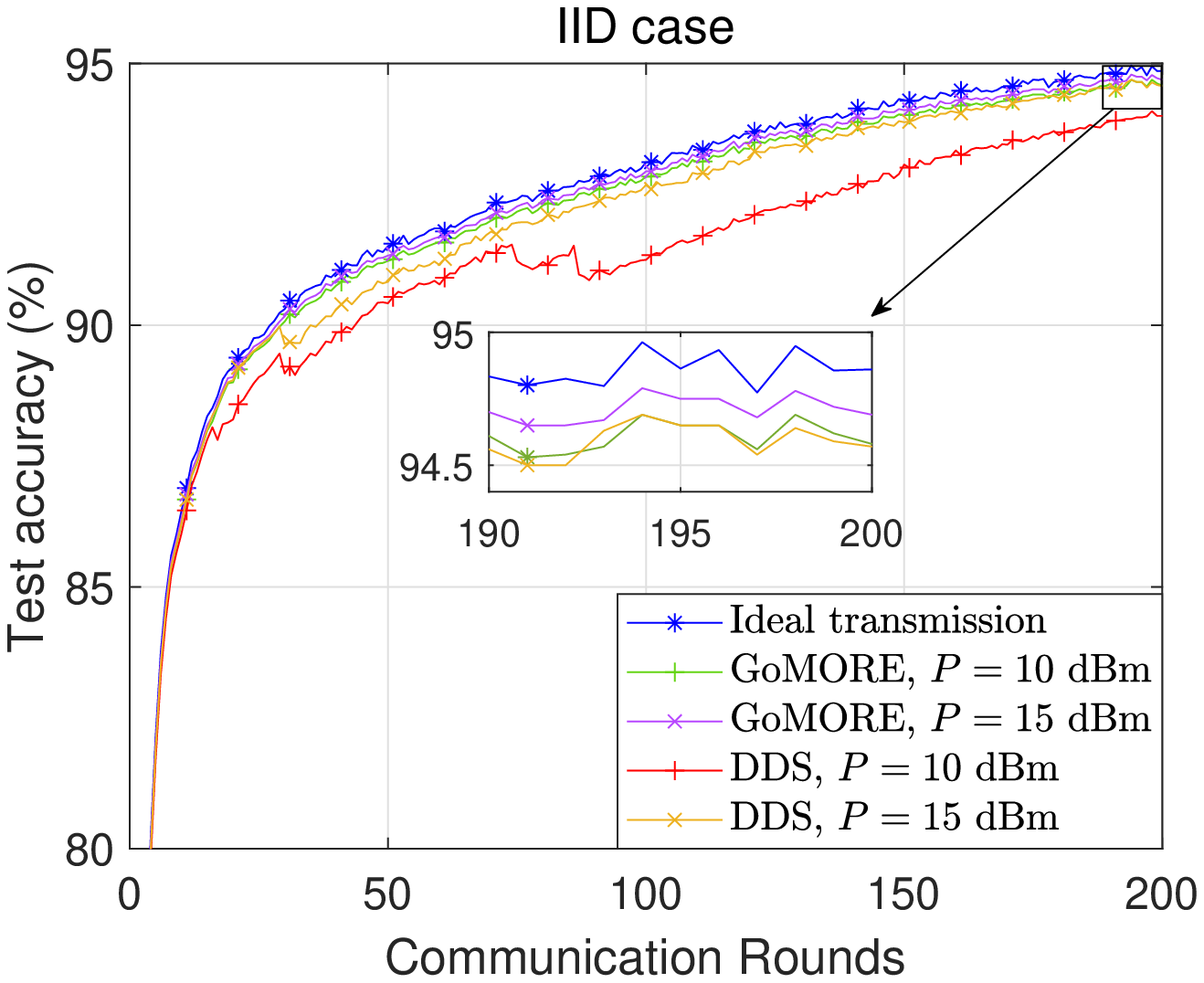}
		\caption{Test accuracy comparison with IID data.}
	\end{minipage}
    \hspace{0.35cm}
	\begin{minipage}[t]{0.25\linewidth}
		\centering
		\includegraphics[width=1\linewidth]{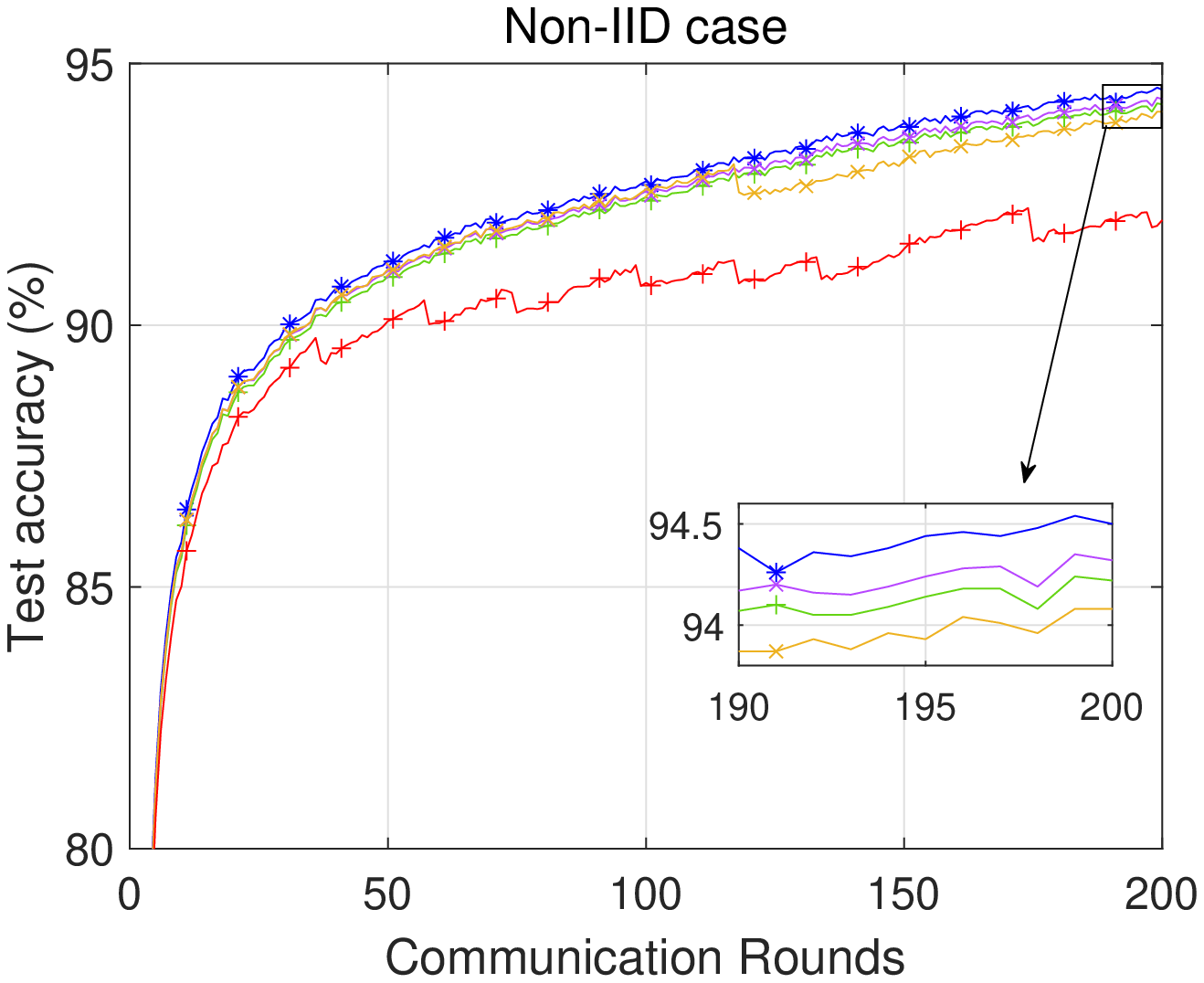}
		\caption{Test accuracy comparison with non-IID data.}
	\end{minipage}
    \hspace{0.35cm}
	\begin{minipage}[t]{0.25\linewidth}
		\centering
		\includegraphics[width=1\linewidth]{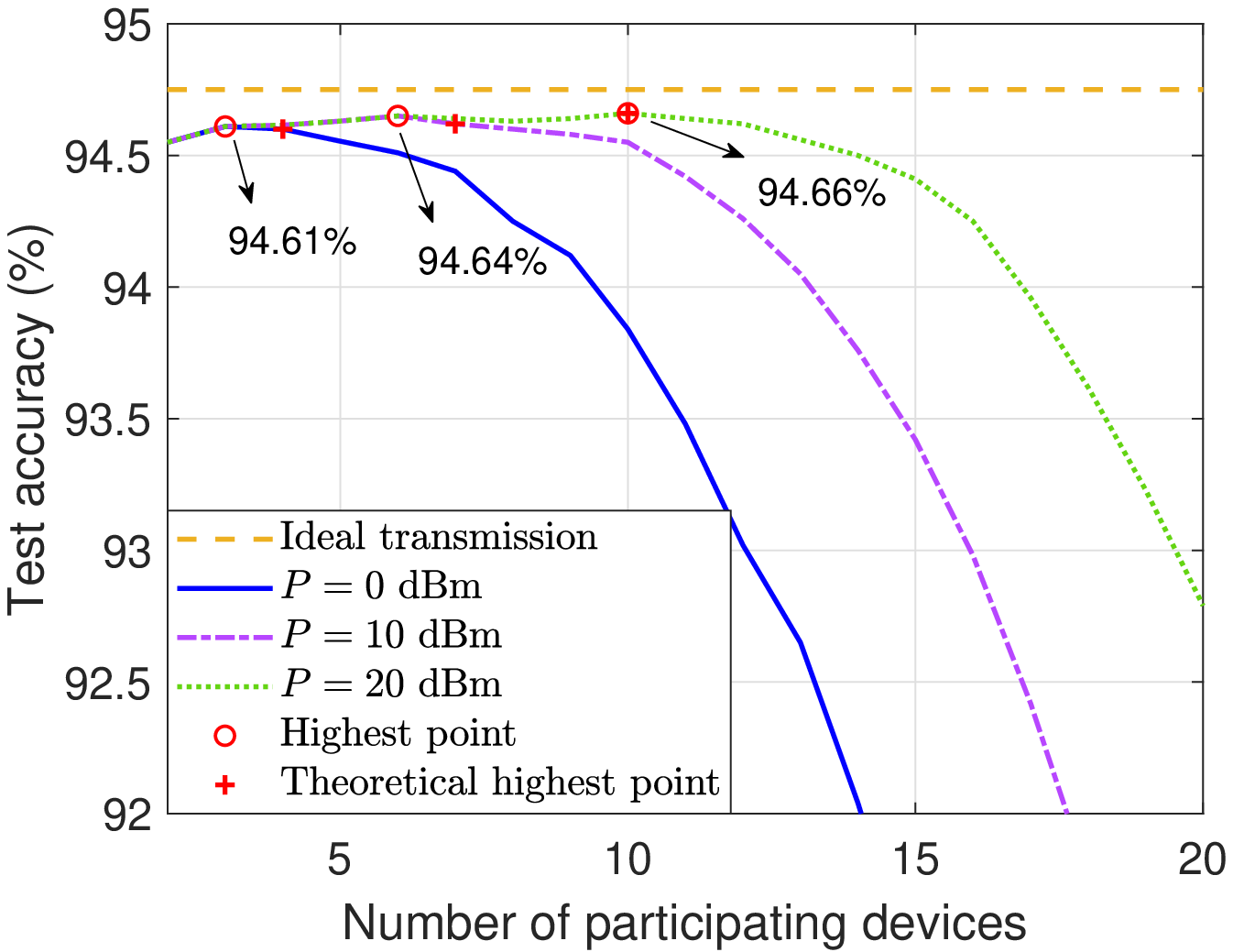}
		\caption{Test accuracy versus number of participating devices.}
	\end{minipage}
 \vspace{-0.5cm}
\end{figure*}

Fig. 4 depicts the test accuracy changes with the number of participating devices. It is observed that, with the increase of $N$, the test accuracy improves first and then decreases. This is because the learning performance is limited first by the number of participating devices and then by unreliable transmission, which unveils the balance between the two factors. In addition, we prefer to activate fewer devices to improve the probability of successful transmission, which is more effective than blindly activating more devices. This is because under limited bandwidth, activating more devices comes with more transmission errors, resulting in even less local models actually exploited for the FL training per round.

\section{Conclusion}
In this paper, we proposed a novel GoMORE strategy for wireless FL with unreliable communications and limited bandwidth. Analytical results revealed that GoMORE always achieves better performance than DDS and there exists an optimal number of participating devices per round, which can be effectively found via exhaustive search. Numerical results confirmed the validity of GoMORE and number optimization.

\appendices
\section{Proof of Lemma 1}
To begin with, we first define the expectations of $\hat{\bm{w}}_{(m+1)T}^k$ and $\hat{\bm{w}}_{m+1}$ over wireless channels and device selection as
\begin{align}\label{a1}
&\bar{\bm{w}}_{(m+1)T}^k \triangleq \mathbb{E}\left [ \hat{\bm{w}}_{(m+1)T}^k \right ]= p_k \bm{w}_{(m+1)T}^k+(1-p_k) \bm{w}_{m}, \nonumber \\
&\bar{\bm{w}}_{m+1} \triangleq \mathbb{E}\left [ \hat{\bm{w}}_{m+1} \right ]= \sum_{k=1}^K \frac{1}{K}\bar{\bm{w}}_{(m+1)T}^k,
\end{align}
respectively. Then, we can express
\begin{align}
\zeta_1&=\mathbb{E}\left [ \Vert \hat{\bm{w}}_{m+1}-\bar{\bm{w}}_{m+1}+\bar{\bm{w}}_{m+1}-\bm{w}_{m+1}  \Vert^2 \right ] \nonumber \\
&\overset{\text{(a)}}{=}\underbrace{\mathbb{E}\left [\Vert \hat{\bm{w}}_{m+1}-\bar{\bm{w}}_{m+1} \Vert^2\right ]}_{A_1} +\underbrace{\mathbb{E}_{\bm{\xi}}\left [\Vert \bar{\bm{w}}_{m+1}-\bm{w}_{m+1}  \Vert^2 \right ]}_{A_2},
\end{align}
where (a) is due to the fact that $\hat{\bm{w}}$ is unbiased, i.e., $\mathbb{E}\left [\hat{\bm{w}}_{m+1}-\bar{\bm{w}}_{m+1} \right ]=0$.
Next, we first rewrite $A_1$ as 
\begin{align}
A_1&
=\mathbb{E} \left [\left \Vert \frac{1}{N}\sum_{k=1}^K \mathbb{I}_k \left ( \hat{\bm{w}}_{(m+1)T}^k -  \bar{\bm{w}}_{m+1}\right) \right \Vert^2 \right ]\nonumber \\
&=\mathbb{E} \left [ \left \Vert \frac{1}{N}\sum_{k=1}^K \mathbb{I}_k \left ( \hat{\bm{w}}_{(m+1)T}^k- \bar{\bm{w}}_{(m+1)T}^k\right) \right. \right. \nonumber \\
&\quad\left.\left. +\frac{1}{N}\sum_{k=1}^K \mathbb{I}_k \left ( \bar{\bm{w}}_{(m+1)T}^k- \bar{\bm{w}}_{m+1}\right)\right\Vert^2\right ]\nonumber \\
&\overset{\text{(a)}}=\frac{1}{N^2}\underbrace{\mathbb{E}\left[\left\Vert \sum_{k=1}^K \mathbb{I}_k \left ( \hat{\bm{w}}_{(m+1)T}^k- \bar{\bm{w}}_{(m+1)T}^k\right) \right\Vert^2\right ]}_{B_1}\nonumber \\
&\quad +\frac{1}{N^2}\underbrace{\mathbb{E}_{\bm{\xi}}\left [\left\Vert \sum_{k=1}^K \mathbb{I}_k \left (\bar{\bm{w}}_{(m+1)T}^k- \bar{\bm{w}}_{m+1}\right) \right\Vert^2\right ]}_{B_2},
\end{align}
where $\mathbb{I}_k=1$ if $k\in \mathcal{S}_m$ and $\mathbb{I}_k=0$ if $k\notin \mathcal{S}_m$. The equality in (a) follows from $\mathbb{E}\left[\hat{\bm{w}}_{(m+1)T}^k \right]=\bar{\bm{w}}_{(m+1)T}^k$. Exploiting the Jensen's Inequality, $B_1$ is bounded by
\begin{align}\label{eq18}
&B_1\leq N \sum_{k=1}^K \mathbb{E}\left [\left \Vert \mathbb{I}_k \left ( \hat{\bm{w}}_{(m+1)T}^k -  \bar{\bm{w}}_{(m+1)T}^k\right) \right \Vert^2 \right]\nonumber \\
&= N \sum_{k=1}^K \Pr(\mathbb{I}_k=1) \mathbb{E}\left [\left \Vert\hat{\bm{w}}_{(m+1)T}^k -  \bar{\bm{w}}_{(m+1)T}^k \right \Vert^2\right ] \nonumber \\
&= \frac{N^2}{K}  \sum_{k=1}^K \mathbb{E}\left [\left \Vert\hat{\bm{w}}_{(m+1)T}^k -  \bar{\bm{w}}_{(m+1)T}^k \right \Vert^2 \right ]\nonumber \\
&\overset{\text{(a)}}{=} \frac{N^2}{K} \sum_{k=1}^K \left(p_k \mathbb{E}_{\bm{\xi}}\left[\left \Vert \bm{w}_{(m+1)T}^k - \bar{\bm{w}}_{(m+1)T}^k \right \Vert^2\right ]\right. \nonumber\\
&\quad\left.+ (1-p_k)\mathbb{E}_{\bm{\xi}}\left[\left \Vert \bm{w}_{m}- \bar{\bm{w}}_{(m+1)T}^k \right \Vert^2\right ]\right)\nonumber \\
& = \frac{N^2}{K} \sum_{k=1}^K  p_k (1-p_k) \mathbb{E}_{\bm{\xi}}\left[ \left \Vert \eta_m \sum_{t=0}^{T-1} \nabla F_k (\bm{w}_{mT+t}^k,\bm{\xi}_{mT+t}^k) \right \Vert^2\right]\nonumber \\
&\overset{\text{(b)}}{\leq}  \frac{N^2}{K} \sum_{k=1}^K p_k (1-p_k) \eta_m^2 T^2 \gamma^2,
\end{align}
where $\Pr(\mathbb{I}_k=1)=\frac{N}{K}$, (a) comes from the definition in (\ref{def1}), and (b) exploits the Jensen's Inequality and \emph{Assumption~1}. Next, we take similar steps as \cite[Appendix B.4]{converge} to bound $B_2$. It follows

\vspace{-0.1cm}
\resizebox{.98\linewidth}{!}{
\begin{minipage}{\linewidth}
\begin{align}\label{a18}
 B_2&\overset{(a)}{\leq} \frac{N(K-N)}{K(K-1)}\sum_{k=1}^K \mathbb{E}_{\bm{\xi}}\left [\left \Vert\bar{\bm{w}}_{(m+1)T}^k- \bar{\bm{w}}_{m+1} \right \Vert^2\right ]\nonumber \\
&\overset{(b)}{=} \frac{N(K-N)}{K(K-1)}\sum_{k=1}^K p_k^2 \mathbb{E}_{\bm{\xi}}\left[\left \Vert \bm{w}_{(m+1)T}^k- \frac{1}{K}\sum_{i=1}^K  \bm{w}_{(m+1)T}^i \right \Vert^2\right]\nonumber \\
&\overset{(c)}{\leq}  \frac{N(K-N)}{K(K-1)}\sum_{k=1}^K p_k^2 \mathbb{E}_{\bm{\xi}}\left [\left \Vert  \bm{w}_{(m+1)T}^k- \bm{w}_{m} \right \Vert^2\right ]\nonumber \\
&\leq  \frac{N(K-N)}{K(K-1)}\sum_{k=1}^K p_k^2\eta_m^2 T^2 \gamma^2,
\end{align}
\end{minipage}}
where (a) comes from the following equalities: $\Pr\{i,j\in \mathcal{S}_m,i\neq j\}=\frac{N(N-1)}{K(K-1)}$, and 

\vspace{-0.2cm}
\noindent
$\sum_i \mathbb{E}_{\bm{\xi}}\left [\left \Vert\bar{\bm{w}}_{(m+1)T}^i- \bar{\bm{w}}_{m+1} \right \Vert^2\right ]+ \sum_{i\neq j}\mathbb{E}_{\bm{\xi}} \left [\left(\bar{\bm{w}}_{(m+1)T}^i- \bar{\bm{w}}_{m+1}  \right)^T\left(\bar{\bm{w}}_{(m+1)T}^j- \bar{\bm{w}}_{m+1} \right)\right ]=0$, (b) exploits the definition in (\ref{a1}), and the inequality in (c) results from $\frac{1}{K} \sum_{k=1}^K  \left( \bm{w}_{(m+1)T}^k- \bm{w}_{m}\right)=\frac{1}{K}\sum_{i=1}^K  \bm{w}_{(m+1)T}^i-\bm{w}_m$ and $\mathbb{E}\left [\left \Vert\bm{x}-\mathbb{E}\{ \bm{x} \}\right \Vert^2\right ] \leq \mathbb{E}\left[\left \Vert\bm{x}\right \Vert^2\right]$.
Then, for $A_2$, we have 
\begin{align}\label{eq20}
A_2 &= \frac{1}{K^2} \mathbb{E}_{\bm{\xi}} \left[\left \Vert \sum_{k=1}^K \left( \bar{\bm{w}}_{(m+1)T}^k -\bm{w}_{(m+1)T}^k\right)\right \Vert^2 \right]\nonumber  \\
&\leq  \sum_{k=1}^K \frac{(1-p_k)^2}{K}\mathbb{E}_{\bm{\xi}} \left[\left \Vert\bm{w}_{(m+1)T}^k- \bm{w}_{m} \right \Vert^2 \right]\nonumber  \\
&\leq \sum_{k=1}^K\frac{(1-p_k)^2}{K} \eta_m^2 T^2 \gamma^2.
\end{align}

Combining the derived bounds in (\ref{eq18})-(\ref{eq20}), we obtain the desired result in (\ref{l1}).

\section{Proof of Lemma 2}
The derivations are analogous to the steps in Appendix A. To be brief, we first have
\begin{align}
\zeta_2 &\leq \frac{1}{N^2}\underbrace{\mathbb{E}\left[ \left \Vert \sum_{k=1}^{K} \mathbb{I}_k \left(\frac{1}{p_k}\tilde{\bm{w}}_{(m+1)T}^k -\bm{w}_{(m+1)T}^k\right) \right \Vert^2\right]}_{C_1} \nonumber \\
&\quad  +\frac{1}{N^2} \underbrace{\mathbb{E}_{\bm{\xi}}\left[\left \Vert \sum_{k=1}^K \mathbb{I}_k \left(\bm{w}_{(m+1)T}^k -\frac{1}{K}\sum_{i=1}^K \bm{w}_{(m+1)T}^i\right)\right \Vert^2\right]}_{C_2}.
\end{align}
Then, by applying the Jensen's Inequality, the definition of $\tilde{\bm{w}}_{(m+1)T}^k$, and \emph{Assumption 2}, $C_1$ is bounded as
\begin{align}
C_1\leq \frac{N^2}{K} \sum_{k=1}^K \frac{1-p_k}{p_k}\mathbb{E}\left[\left\Vert\bm{w}_{(m+1)T}^k \right \Vert^2\right]\leq \frac{N^2}{K} \sum_{k=1}^K  \frac{1-p_k}{p_k}G^2.
\end{align}
Moreover, it is easy to find that $C_2$ has been bounded in (\ref{a18}). Combining all these derived bounds, we complete the proof.

\end{document}